\documentclass[preprint,superscriptaddress,
showpacs,preprintnumbers,amsmath,amssymb]{revtex4}

\usepackage{amssymb,amsmath,amsthm}
\usepackage{epsfig}

\newcommand{\ket}[1]{{\left| #1 \right\rangle}}
\newtheorem{Thm}{Theorem}
\newtheorem{Cor}[Thm]{Corollary}
\newtheorem{Lem}[Thm]{Lemma}
\newtheorem{Prop}[Thm]{Proposition}

\begin{document}
%%%%%%%%%%%%%%%%%%%%%%%%%%%%%%%%%%%%%%%%%%%%%%%%%%%%%%%%%%%%%%%%%%%%%%%%%%
%                                                                        %
%                                 Title                                  %
%                                                                        %
%%%%%%%%%%%%%%%%%%%%%%%%%%%%%%%%%%%%%%%%%%%%%%%%%%%%%%%%%%%%%%%%%%%%%%%%%%
\title{Forgeable quantum messages
in arbitrated quantum signature schemes}

\author{Taewan Kim}%\email{TaewanKim@etri.re.kr}
\affiliation{
 Electronics and Telecommunications Research Institute, 
 Daejeon 34129, Korea
}
\author{Hyang-Sook Lee}%\email{hsl@ewha.ac.kr}
\affiliation{
 Department of Mathematics,
 Ewha Womans University, Seoul 03760, Korea
}
\author{Soojoon Lee}\email{level@khu.ac.kr}
\affiliation{
 Department of Mathematics and Research Institute for Basic Sciences,
 Kyung Hee University, Seoul 02447, Korea
}
\date{\today}

%%%%%%%%%%%%%%%%%%%%%%%%%%%%%%%%%%%%%%%%%%%%%%%%%%%%%%%%%%%%%%%%%%%%%%%%%%
%                                                                        %
%                              Abstract                                  %
%                                                                        %
%%%%%%%%%%%%%%%%%%%%%%%%%%%%%%%%%%%%%%%%%%%%%%%%%%%%%%%%%%%%%%%%%%%%%%%%%%
\begin{abstract}
Recently, the concept on `forgeable quantum messages' in arbitrated quantum signature schemes 
was introduced by T.~Kim {\it et al.}~[Phys. Scr., {\bf 90}, 025101 (2015)], 
and it has been shown that 
there always exists such a forgeable quantum message 
for every known arbitrated quantum signature scheme 
with four quantum encryption operators and the specific two rotation operators. 
We first extend the result to the case of any two unitary rotation operators,
and then consider the forgeable quantum messages in the schemes 
with four quantum encryption operators and three or more rotation operators.
We here present a necessary and sufficient condition for existence of a forgeable quantum message,
and moreover, by employing the condition, 
show that there exists an arbitrated quantum signature scheme 
which contains no forgeable quantum message-signature pairs.
\end{abstract}

%\keywords{Entanglement purification protocol, Controlled-NOT operation}
\maketitle

%%%%%%%%%%%%%%%%%%%%%%%%%%%%%%%%%%%%%%%%%%%%%%%%%%%%%%%%%%%%%%%%%%%%%%
%%%                                                                %%%
%%%                         Introduction                           %%%
%%%                                                                %%%
%%%%%%%%%%%%%%%%%%%%%%%%%%%%%%%%%%%%%%%%%%%%%%%%%%%%%%%%%%%%%%%%%%%%%%
\section{Introduction}
\label{sec:intro}
It has been known that 
quantum cryptography provides us 
with unconditional security in key distribution.  
In order to extend the range of 
quantum cryptographic primitives, 
other quantum-mechanics-based cryptographic schemes such as signature schemes
should be studied and developed.

Even though quantum mechanics can help in digitally signing classical messages~\cite{GC}, 
and its techniques have been developed~\cite{WCR,WMC}, 
it unfortunately turns out that 
it is not possible to digitally sign quantum messages~\cite{BCGST}.
On the other hand, it has been shown that 
it could be possible to quantumly sign quantum messages 
with the assistance of an arbitrator, 
which is called the arbitrated quantum signature 
(AQS)~\cite{ZK,LCL,CM,ZQ,GQGW,CCH,LLWZC,ZZL,SL,ZQSSS,ZLS,KCJL}. 

However, it has been also shown that 
the AQS schemes can be vulnerable to 
receiver's forgery attacks~\cite{GQGW,CCH,ZZL,ZQSSS,ZLS,KCJL}.
In particular, according to the result in Ref.~\cite{KCJL}, 
for almost all known AQS schemes, 
there are forgeable quantum message-signature pairs
which can be forged by the receiver, 
although all other pairs cannot be forged. 
Furthermore, there has not yet been an analytical proof 
that there exists an AQS scheme 
without any forgeable quantum message-signature pairs.

In this paper, we first show that 
there is always a forgeable quantum message-signature pair
in an AQS scheme with two quantum random rotations, 
and present a necessary and sufficient condition 
for a given AQS scheme with four quantum encryption operators 
and three or more quantum random rotations 
to produce at least one forgeable quantum message. 
By using this condition, we analytically prove that 
there exists an AQS scheme 
without any forgeable quantum message-signature pairs.

%%%%%%%%%%%%%%%%%%%%%%%%%%%%%%%%%%%%%%%%%%%%%%%%%%%%%%%%%%%%%%%%%%%%%%
%%%                                                                %%%
%%%                         Main results                           %%%
%%%                                                                %%%
%%%%%%%%%%%%%%%%%%%%%%%%%%%%%%%%%%%%%%%%%%%%%%%%%%%%%%%%%%%%%%%%%%%%%%
\section{Forgeable quantum messages}
\label{sec:2}
Most qubit-based AQS schemes have two quantum signature operators, 
the random rotation $\{R_j\}_{j\in\mathbb{Z}_2}$ defined by
two Pauli operators $\sigma_x$ and $\sigma_z$, that is,
$R_{0}=\sigma_x$ and $R_{1}=\sigma_z$,
and the quantum encryption $\{E_k\}_{k\in\mathbb{Z}_4}$~\cite{BR} 
with unitary operator $E_k$
satisfying 
\begin{equation}
\frac{1}{4}\sum_{k\in\mathbb{Z}_4}E_{k}\rho E_{k}^\dagger = \frac{1}{2}I
\label{eq:E_k}
\end{equation}
for any qubit state $\rho$.
It follows that
\begin{equation}
\frac{1}{8}\sum_{j\in\mathbb{Z}_2}\sum_{k\in\mathbb{Z}_4} 
E_{k}R_{j}\rho R_{j}^\dagger E_{k}^\dagger = \frac{1}{2}I.
\label{eq:E_kR_j}
\end{equation}
Thus, for a given quantum message $\ket{M}$ and a previously shared key $(j,k)$, 
its signature $\ket{S}$ becomes
\begin{equation}
\ket{S}=E_k R_j \ket{M}.
\label{eq:Sign}
\end{equation}
Let $\left(\ket{M'},\ket{S'}\right)$ be the transmitted message-signature pair. 
Then we say that the signature is valid if
\begin{equation}
\ket{M'}\simeq R_j^{\dagger}E_k^{\dagger}\ket{S'},
\label{eq:valid}
\end{equation}
where $\ket{A}\simeq \ket{B}$ if and only if $\ket{A}$ equals $\ket{B}$ up to global phase.
In other words, 
\begin{equation}
\ket{M'}=e^{i\theta_{jk}} R_j^{\dagger}E_k^{\dagger}\ket{S'}
\label{eq:valid_gp}
\end{equation}
for some real number $\theta_{jk}$.
We note that
one can determine with high probability 
whether the signature is valid or not,
by using the swap test~\cite{BCWW} 
for proper number of copies of the states.

In addition, we remark that 
quantum encryptions in all existing AQS schemes
can be essentially reduced to the quantum encryption with four unitary operators 
\begin{equation}
E_k =\sigma_k W
\label{eq:q_encryption}
\end{equation}  
for some uniatry operator $W$, 
which is called 
an {\em assistant unitary operator} of the AQS scheme~\cite{ZLS},
where $\sigma_{0}=I$, $\sigma_{1}=\sigma_{x}$, $\sigma_{2}=\sigma_{y}$ and $\sigma_{3}=\sigma_{z}$.

Recently, Kim {\it et al.}~\cite{KCJL} proposed a new concept, {\em forgeable quantum messages} as follows.
For a given AQS scheme 
with the unitary random rotation $\{R_j\}$ and 
quantum encryption of the form in Eq.~(\ref{eq:q_encryption}),
there exist a quantum message $\ket{M_0}$,
a non-identity unitary $Q$ and a unitary $U$
such that
\begin{equation}
R_j^{\dagger}W^{\dagger}\sigma_k^{\dagger}Q\sigma_{k}WR_{j}\ket{M_0}\simeq U\ket{M_0}
\label{eq:forgeable}
\end{equation}
for all $j\in\mathbb{Z}_2$ and $k\in\mathbb{Z}_4$ 
if and only if 
the quantum message $\ket{M_0}$ 
is said to be {\em forgeable} in the AQS scheme.
This implies that
the receiver can forge at least one quantum message and its signature 
by exploiting a pair of forgery operators $Q$ and $U$, 
although all other quantum message-signature pairs cannot be forged.

\section{Main results}\label{sec:3}
\begin{Thm}\label{Thm1}
Assume that an AQS scheme consists of
the random rotation $\{R_j\}_{j\in\mathbb{Z}_2}$ defined by two unitary operators
and a quantum encryption $\{\sigma_k W\}_{k\in\mathbb{Z}_4}$
with an assistant unitary operator $W$.
Then there exists at least one forgeable qubit message $\ket{M_0}$ 
in the AQS scheme.
%that is,
%there exist a qubit message $\ket{M_0}$ and forgery unitary operators $Q$ and $U$
%satisfying the condition in~(\ref{eq:forgeable}) for all $j\in\mathbb{Z}_2$ and $k\in\mathbb{Z}_4$.
\end{Thm}
\begin{proof}
We here take $Q=\sigma_1$. 
Then for any unitary $W$, we note that
\begin{eqnarray}
R_{0}^{\dagger}W^{\dagger}\sigma_{k}\sigma_{1}\sigma_{k}WR_{0}
&\simeq&R_{0}^{\dagger}W^{\dagger}\sigma_{1}WR_{0},\nonumber\\
R_{1}^{\dagger}W^{\dagger}\sigma_{k}\sigma_{1}\sigma_{k}WR_{1}
&\simeq&R_{1}^{\dagger}W^{\dagger}\sigma_{1}WR_{1}
\label{eq:simeq_prop}
\end{eqnarray}
for all $k\in\mathbb{Z}_4$,
since 
%$\sigma_1$ commutes or anti-commutes with all Pauli matrices, that is,
$\sigma_{1}\simeq \sigma_{k}\sigma_{1}\sigma_{k}$ for all $k\in\mathbb{Z}_4$.
Thus it suffices to show that 
\begin{equation}
R_{1}^{\dagger}W^{\dagger}\sigma_{1}WR_{1}\ket{M}
\simeq R_{0}^{\dagger}W^{\dagger}\sigma_{1}WR_{0}\ket{M}
\label{eq:twoRotation0}
\end{equation}
for some qubit state $\ket{M}$.

Since the unitary operator
\begin{equation}
\left(R_{0}^{\dagger}W^{\dagger}\sigma_{1}WR_{0}\right)^{\dagger}
\left(R_{1}^{\dagger}W^{\dagger}\sigma_{1}WR_{1}\right)
\end{equation} 
has an eigenstate with a complex eigenvalue of modulus one
by the spectral decomposition theorem,
there exists a message $\ket{M_0}$ such that
\begin{equation}
\left(R_{0}^{\dagger}W^{\dagger}\sigma_{1}WR_{0}\right)^{\dagger}
\left(R_{1}^{\dagger}W^{\dagger}\sigma_{1}WR_{1}\right) \ket{M_0}
\simeq \ket{M_0},
\label{eq:eigenvector}
\end{equation}
that is,
\begin{equation}
R_{1}^{\dagger}W^{\dagger}\sigma_{1}WR_{1}\ket{M_0}
\simeq R_{0}^{\dagger}W^{\dagger}\sigma_{1}WR_{0}\ket{M_0}.
\label{eq:twoRotation}
\end{equation}
\end{proof}
For example, we assume that an AQS scheme consists of
the random rotation $\{R_j\}_{j\in\mathbb{Z}_2}$ defined by two Pauli operators ${\sigma}_x$ and ${\sigma}_z$ and a quantum encryption $\{\sigma_k W\}_{k\in\mathbb{Z}_4}$
with an assistant unitary operator $W$.
Then there exists at least one forgeable qubit message $\ket{M_0}$
by Theorem~\ref{Thm1}.
It means that there exists at least one forgeable qubit message
in most of the previously known AQS schemes.

We now introduce the following basic property in linear algebra~\cite{Sh} 
to obtain our main results.
\begin{Prop}\label{Prop2}
Two $r\times r$ matrices $A$ and $B$ have a common eigenvector
if and only if 
\begin{eqnarray}
\bigcap_{k,l=1}^{r-1}\ker{(A^kB^l-B^lA^k)}\neq\{0\}, 
\label{eq:Prop2}
\end{eqnarray}
where, for an $r\times r$ matrix $C$, $\ker(C)=\{x\in\mathbb{C}^r|Cx=0\}$ 
is called the kernel of $C$.
\end{Prop}
We remark that $\ker(C)=\{0\}$ if and only if $C$ is invertible, 
and hence two $2\times 2$ matrices $A$ and $B$ have a common eigenvector 
if and only if 
$\ker(AB-BA)\neq\{0\}$, that is, $\det{(AB-BA)}=0$. 
Since we here deal with only $2\times 2$ matrices, 
we obtain the following lemma 
by employing the above remark and Proposition~\ref{Prop2}.
\begin{Lem}\label{Lem1}
Assume that an AQS scheme consists of
the random rotation $\{\tilde{R}_j\}_{j\in\{1,2,3\}}$ with $\tilde{R}_j\in\{\sigma_l: l\in\mathbb{Z}_4\}$ 
and a quantum encryption $\{\sigma_k W\}_{k\in\mathbb{Z}_4}$
with an assistant unitary operator $W$, 
and that there exists a forgeable quantum message, 
that is, there exist a quantum message $\ket{M_0}$, 
forgery operators $Q$ and $U$ such that
\begin{equation}
\tilde{R}_j^{\dagger}W^{\dagger}\sigma_k^{\dagger}Q\sigma_{k}W\tilde{R}_{j}\ket{M_0}\simeq U\ket{M_0}
\label{eq:forge_Q},
\end{equation}
where 
\begin{equation}
Q=q_0 \sigma_0+iq_1 \sigma_1-iq_2 \sigma_2 +iq_3 \sigma_3
\label{eq:Q}
\end{equation}
for $q_l \in\mathbb{R}$, $q_0\ge 0$ with $\sum_{l\in\mathbb{Z}_4}q_l^2 =1$. 
Then $q_l = q_{l'}=0$ for some distinct $l$ and $l'$ in $\mathbb{Z}_4$.
\end{Lem}
\begin{proof}
Without loss of generality, we may assume that 
for each $j\in\{1,2,3\}$, 
$\tilde{R}_j=\sigma_j$. 
Then since
\begin{equation}
\sigma_j^{\dagger}W^{\dagger}\sigma_k^{\dagger}Q\sigma_{k}W\sigma_{j}\ket{M_0}\simeq U\ket{M_0}
\label{eq:forgeable2}
\end{equation}
for all $j\in\{1,2,3\}$ and $k\in\mathbb{Z}_4$, 
it follows that
\begin{equation}
\sigma_{j'}^{\dagger}W^{\dagger}\sigma_{k'}^{\dagger}Q^{\dagger}\sigma_{k'}W\sigma_{j'}
\sigma_j^{\dagger}W^{\dagger}\sigma_k^{\dagger}Q\sigma_{k}W\sigma_{j}
\ket{M_0}\simeq\ket{M_0}
\label{eq:forgeable3}
\end{equation}
for all $j, j'\in\{1,2,3\}$ and $k, k'\in\mathbb{Z}_4$. 
In other words, 
for all $j, j'\in\{1,2,3\}$ and $k, k'\in\mathbb{Z}_4$, 
all $2\times 2$ matrices
\begin{equation}
\sigma_{j'}^{\dagger}W^{\dagger}\sigma_{k'}^{\dagger}Q^{\dagger}\sigma_{k'}W\sigma_{j'}
\sigma_j^{\dagger}W^{\dagger}\sigma_k^{\dagger}Q\sigma_{k}W\sigma_{j}
\label{eq:forgeable_matrices}
\end{equation}
have  a common eigenvector $\ket{M_0}$. 
In particular, 
the two $2\times 2$ matrices
\begin{eqnarray}
\sigma_1^{\dagger}W^{\dagger}\sigma_0^{\dagger}Q^{\dagger}\sigma_{0}W\sigma_{1}
\sigma_1^{\dagger}W^{\dagger}\sigma_1^{\dagger}Q\sigma_{1}W\sigma_{1}
&=& 
\sigma_1W^{\dagger}Q^{\dagger}\sigma_1Q\sigma_{1}W\sigma_{1}
\nonumber\\
\sigma_1^{\dagger}W^{\dagger}\sigma_0^{\dagger}Q^{\dagger}\sigma_{0}W\sigma_{1}
\sigma_1^{\dagger}W^{\dagger}\sigma_2^{\dagger}Q\sigma_{2}W\sigma_{1}
&=& 
\sigma_1W^{\dagger}Q^{\dagger}\sigma_2Q\sigma_{2}W\sigma_{1}
\label{eq:forgeable_matrices2}
\end{eqnarray}
have a common eigenvector $\ket{M_0}$. 
Thus, by Proposition~\ref{Prop2}, 
\begin{eqnarray}
0&=&
\det\left(\sigma_1W^{\dagger}Q^{\dagger}\sigma_1Q\sigma_{1}Q^{\dagger}\sigma_2Q\sigma_{2}W\sigma_{1}
-\sigma_1W^{\dagger}Q^{\dagger}\sigma_2 Q\sigma_{2}Q^{\dagger}\sigma_1Q\sigma_{1}W\sigma_{1} \right)
\nonumber\\
&=& 
q_0^2q_1^2q_2^2+q_0^2q_2^2q_3^2+q_0^2q_1^2q_3^2+q_1^2q_2^2q_3^2. 
\label{eq:forgeable_det} 
\end{eqnarray}
Hence we obtain that there exist at least two $l$, $l'$ in $\mathbb{Z}_4$ 
such that $q_l = q_{l'}=0$.
\end{proof}

We note that any $2\times 2$ unitary operator can be expressed as
the form of Eq.~(\ref{eq:Q}) up to global phase.
Thus we may assume that an assistant unitary operator $W$ is 
\begin{eqnarray}
W &=& w_0 \sigma_0+iw_1 \sigma_1-iw_2 \sigma_2 +iw_3 \sigma_3,
\label{eq:W}
\end{eqnarray}
where $w_l \in\mathbb{R}$, $w_0\ge 0$
and $\sum_{l\in\mathbb{Z}_4}w_l^2 =1$.
Then for distinct $l$, $m$, $n \in \{1,2,3\}$,
let
\begin{eqnarray}
\alpha_l
&\equiv&w_{0}^{2}+w_{l}^{2}-\frac{1}{2}
=\frac{1}{2}-w_{m}^{2}-w_{n}^{2},
\nonumber \\
\beta_m&\equiv&w_{0}w_{m}+w_{n}w_{l},
\nonumber \\
\gamma_n&\equiv&w_{0}w_{n}-w_{l}w_{m},
\label{eq:alphabetagamma}
\end{eqnarray}
and let $\mathcal{W}_{lmn}$ be the set of assistant unitary operators satisfying 
the equality $\alpha_l \beta_m \gamma_n =0$, that is, 
\begin{equation}
\mathcal{W}_{lmn}=\{W:\alpha_l \beta_m \gamma_n =0 \}. 
\label{eq:calW}
\end{equation}
%Let $\mathcal{W}$ be the union of $\mathcal{W}_{jkl}$'s 
%for all distinct $j$, $k$, $l \in \mathbb{Z}_4\setminus \{0\}$.
\begin{Thm}\label{Thm2}
Assume that an AQS scheme consists of
the random rotation $\{\tilde{R}_j\}_{j\in\{1,2,3\}}$ with  $\tilde{R}_j\in\{\sigma_k: k\in\mathbb{Z}_4\}$
and a quantum encryption $\{\sigma_k W\}_{k\in\mathbb{Z}_4}$
with an assistant unitary operator $W$.
Then there exists at least one forgeable quantum message
if and only if 
\begin{equation}
W\in \mathcal{W}_{lmn}
\end{equation}
for some distinct $l$, $m$, $n \in \{1,2,3\}$.
\end{Thm}

\begin{proof}
%(Sufficient condition) 
As in the proof of Lemma~\ref{Lem1}, 
without loss of generality, we may assume that 
$\tilde{R}_j=\sigma_j$ for every $j\in\{1,2,3\}$.  
Suppose that $W\in\mathcal{W}_{lmn}$ 
for some distinct $l$, $m$, $n \in \{1,2,3\}$. 
We may also assume that $W\in\mathcal{W}_{123}$, 
without loss of generality, 
that is, $\alpha_1\beta_2\gamma_3=0$.
%If $\alpha_1=0$ then 
%If $Q=\sigma_1$ then
Note that for any unitary $W$, 
\begin{eqnarray}
\sigma_{1}^{\dagger}W^{\dagger}\sigma_{1}W\sigma_{1}
&\simeq&\sigma_{1}^{\dagger}W^{\dagger}\sigma_{k}\sigma_{1}\sigma_{k}W\sigma_{1}\nonumber\\
&\simeq& 
  \left(
    \begin{array}{cc}
      -2\beta_2 & 2\alpha_1+2i\gamma_3 \\
      2\alpha_1-2i\gamma_3 & 2\beta_2 \\
    \end{array}
  \right),\nonumber\\
\sigma_{2}^{\dagger}W^{\dagger}\sigma_{1}W\sigma_{2}
&\simeq&\sigma_{2}^{\dagger}W^{\dagger}\sigma_{k}\sigma_{1}\sigma_{k}W\sigma_{2}\nonumber\\
&\simeq& 
  \left(
    \begin{array}{cc}
      -2\beta_2 & -2\alpha_1-2i\gamma_3 \\
      -2\alpha_1+2i\gamma_3 & 2\beta_2 \\
    \end{array}
  \right),
\nonumber\\
\sigma_{3}^{\dagger}W^{\dagger}\sigma_{1}W\sigma_{3}
&\simeq&\sigma_{3}^{\dagger}W^{\dagger}\sigma_{k}\sigma_{1}\sigma_{k}W\sigma_{3}
\nonumber\\
&\simeq& 
  \left(
    \begin{array}{cc}
      2\beta_2 & -2\alpha_1+2i\gamma_3 \\
      -2\alpha_1-2i\gamma_3 & -2\beta_2 \\
    \end{array}
  \right)
\label{eq:simeq_prop}
\end{eqnarray}
for all $k\in\mathbb{Z}_4$.
Therefore, 
%\begin{equation}
%\alpha_1=0 \text{ or } \beta_2=0 \text{ or } \gamma_3=0,
%\end{equation}
%$\alpha_1=0 \text{ or } \beta_2=0 \text{ or } \gamma_3=0$,
%then 
it follows from Eqs.~(\ref{eq:simeq_prop}) that
for each $j\in\{1,2,3\}$ and $k\in\mathbb{Z}_4$, 
the matrix 
\begin{equation}
\sigma_{j}^{\dagger}W^{\dagger}\sigma_{1}W\sigma_{j}
\simeq\sigma_{j}^{\dagger}W^{\dagger}\sigma_{k}\sigma_{1}\sigma_{k}W\sigma_{j}
\label{eq:all_forge}
\end{equation}
is one of at most two matrices up to global phase, 
since $\alpha_1\beta_2\gamma_3=0$. 
For example, 
if $\alpha_1=0$ then 
\begin{eqnarray}
\sigma_{1}^{\dagger}W^{\dagger}\sigma_{1}W\sigma_{1}
&\simeq& 
  \left(
    \begin{array}{cc}
      -2\beta_2 & 2i\gamma_3 \\
      -2i\gamma_3 & 2\beta_2 \\
    \end{array}
  \right),\nonumber\\
\sigma_{2}^{\dagger}W^{\dagger}\sigma_{1}W\sigma_{2}
&\simeq& 
  \left(
    \begin{array}{cc}
      -2\beta_2 & -2i\gamma_3 \\
      2i\gamma_3 & 2\beta_2 \\
    \end{array}
  \right)
\simeq 
  \left(
    \begin{array}{cc}
      2\beta_2 & 2i\gamma_3 \\
      -2i\gamma_3 & -2\beta_2 \\
    \end{array}
  \right)\simeq
  \sigma_{3}^{\dagger}W^{\dagger}\sigma_{1}W\sigma_{3}.
\label{eq:example_simeq_prop}
\end{eqnarray}
Hence we can easily show that there exists a forgeable quantum message, 
as in the proof of Theorem~\ref{Thm1}.

We now suppose that there exists a forgeable quantum message.
For a given assistant unitary operator $W$,
there exist a quantum message $\ket{M_0}$,
a non-identity unitary $Q=q_0 \sigma_0+iq_1 \sigma_1-iq_2 \sigma_2 +iq_3 \sigma_3$ 
and a unitary $U$
such that
\begin{equation}
\sigma_j^{\dagger}W^{\dagger}\sigma_k^{\dagger}Q\sigma_{k}W\sigma_{j}\ket{M_0}\simeq U\ket{M_0}
\label{eq:forgeable2}
\end{equation}
for all $j\in\{1,2,3\}$ and $k\in\mathbb{Z}_4$. 
Then it follows from Lemma~\ref{Lem1} that 
$q_l = q_{l'} = 0$ for some distinct $l, l'\in\mathbb{Z}_4$.
Then we have either $Q=q_0 \sigma_0+i(-1)^{t-1} q_t \sigma_t$ 
or $Q=i(-1)^{t-1} q_t \sigma_t+i(-1)^{s-1} q_{s} \sigma_{s}$ 
for distinct $t$ and $s$ in $\{1,2,3\}$.

If  $Q=q_0 \sigma_0+i(-1)^{t-1} q_t \sigma_t$ then 
we can obtain from tedious but straightforward calculations 
that $\alpha_j\beta_k\gamma_l=0$ for distinct $j, k, l\in\{1,2,3\}$.
For example,  
if $Q=q_0 \sigma_0+i q_1 \sigma_1$ then,
as in the proof of Lemma~\ref{Lem1},
the two $2\times 2$ matrices 
\begin{eqnarray}
\sigma_1W^{\dagger}\sigma_0Q^{\dagger}\sigma_{0}W\sigma_{1}
\sigma_2W^{\dagger}\sigma_0Q\sigma_{0}W\sigma_{2}
&=& 
\sigma_1W^{\dagger}Q^{\dagger}W\sigma_{1}
\sigma_2W^{\dagger}QW\sigma_{2},
\nonumber\\
\sigma_1W^{\dagger}\sigma_0Q^{\dagger}\sigma_{0}W\sigma_{1}
\sigma_3W^{\dagger}\sigma_0Q\sigma_{0}W\sigma_{3}
&=& 
\sigma_1W^{\dagger}Q^{\dagger}W\sigma_{1}
\sigma_3W^{\dagger}QW\sigma_{3}
\label{eq:forgeable_matrices2}
\end{eqnarray}
have a common eigenvector, that is, 
the two matrices 
\begin{eqnarray}
&&\left(
    \begin{array}{cc}
      q_0-2iq_1\beta_2 & -2q_1i\gamma_3+2iq_1\alpha1\\
      2q_1\gamma_3+2iq_1\alpha_1 & q_0+2iq_1\beta_2 \\
    \end{array}
  \right)^{\dagger}
  \left(
    \begin{array}{cc}
      q_0-2iq_1\beta_2 & 2q_1\gamma_3-2iq_1\alpha1\\
      -2q_1\gamma_3-2iq_1\alpha_1 & q_0+2iq_1\beta_2 \\
    \end{array}
  \right), \nonumber\\
&&\left(
    \begin{array}{cc}
      q_0-2iq_1\beta_2 & -2q_1i\gamma_3+2iq_1\alpha1\\
      2q_1\gamma_3+2iq_1\alpha_1 & q_0+2iq_1\beta_2 \\
    \end{array}
  \right)^{\dagger}
  \left(
    \begin{array}{cc}
      q_0+2iq_1\beta_2 & -2q_1\gamma_3-2iq_1\alpha1\\
      2q_1\gamma_3-2iq_1\alpha_1 & q_0-2iq_1\beta_2 \\
    \end{array}
  \right) 
\label{eq:forgeable_matrices3}
\end{eqnarray}
have a common eigenvector.
It follows from Proposition~\ref{Prop2} and tedious calculations that 
$\alpha_1\beta_2\gamma_3=0$, that is, 
$W\in \mathcal{W}_{123}$.

If  $Q=i(-1)^{t-1} q_t \sigma_t+i(-1)^{s-1} q_{s} \sigma_{s}$ then 
%from tedious but straightforward calculations 
in the same way as above, 
we can show 
that $\alpha_3\beta_{2}\gamma_1=0$.
In particular,  
if $Q=iq_1 \sigma_1-i q_2 \sigma_2$ then 
it can be obtained that 
$\alpha_3=0$ or $\beta_{2}^2+\gamma_1^2=0$.
Therefore, we conclude that
$W\in \mathcal{W}_{321}$.
\end{proof} 

Theorem~\ref{Thm2} provides us with a perfect characterization of the AQS schemes
with forgeable quantum messages 
which we have dealt with in this paper. 
For example, if an assistant unitary operator in an AQS scheme 
is 
$H=\left(\sigma_0+i\sigma_1-i\sigma_2 +i\sigma_3\right)/2$ 
in~\cite{CCH} or 
$W_a\simeq\left(i\sigma_1-i\sigma_2 +i\sqrt{2}\sigma_3\right)/2$ 
in~\cite{ZLS,KCJL}, 
then $\alpha_1=0$ or $\alpha_3=0$, respectively, 
and hence $H\in\mathcal{W}_{123}$ and $W_a\in\mathcal{W}_{312}$. 
Therefore we can clearly obtain from Theorem~\ref{Thm2} that 
an AQS scheme which uses $H$ or $W_a$ as an assistant unitary operator 
has at least one forgeable quantum message. 

In Ref.~\cite{KCJL}, 
the AQS scheme with three or more random rotation operators 
and an assistant unitary operator 
\begin{equation}
T=\frac{i\sigma_1-i\sigma_2 +i\sigma_3}{\sqrt{3}}
\label{eq:T}
\end{equation}
was considered as 
a good candidate for the AQS scheme without forgeable quantum messages, 
and it was numerically shown that 
the AQS scheme does not generate any forgeable quantum message, 
but any analytical proof has not yet been obtained. 
However, it is readily shown that $T\notin\mathcal{W}_{lmn}$
for any distinct $l$, $m$, $n \in \{1,2,3\}$, and 
furthermore $(\sigma_0+(-1)^{t-1} i\sigma_t +(-1)^{s-1} i\sigma_s)/\sqrt{3} \notin\mathcal{W}_{lmn}$ 
for any distinct $t$, $s \in \{1,2,3\}$ and any distinct $l$, $m$, $n \in \{1,2,3\}$.
It follows that 
we can use Theorem~\ref{Thm2} to simply prove the following corollary. 
\begin{Cor}\label{Cor}
There exists an AQS scheme without any forgeable quantum message.
\end{Cor}

We remark that the set 
\begin{equation} 
\mathcal{W}=\left\{W: W\notin\mathcal{W}_{lmn} 
~\mathrm{for~all}~\{l, m, n\} = \{1,2,3\}\right\}
\label{eq:setW}
\end{equation}
is an infinite set, 
and hence 
there exist infinitely many AQS schemes with no forgeable quantum message.

%%%%%%%%%%%%%%%%%%%%%%%%%%%%%%%%%%%%%%%%%%%%%%%%%%%%%%%%%%%%%%%%%%%%%%
%%%                                                                %%%
%%%                         Conclusion                             %%%
%%%                                                                %%%
%%%%%%%%%%%%%%%%%%%%%%%%%%%%%%%%%%%%%%%%%%%%%%%%%%%%%%%%%%%%%%%%%%%%%%
\section{Conclusion}
In this paper, we have dealt with a recently proposed concept, 
{\em forgeable quantum messages}, in the AQS schemes, 
and have proved that 
there exists a forgeable quantum message 
in every AQS scheme with two quantum random rotations.
Furthermore, we have presented 
a necessary and sufficient condition 
for an AQS scheme with four quantum encryption operators 
and three or more quantum random rotations 
in which there exists a forgeable quantum message.
In other words,  
we have perfectly characterized assistant unitary operators in 
the quantum encryption operators which produce at least one forgeable quantum message.
In the sequel, we have analytically shown that 
there exists an AQS scheme 
without any forgeable quantum messages,
by exploiting the necessary and sufficient condition. 

In addition to the forgery problem, 
the AQS schemes may have other security problems 
which we have not considered in this paper, 
but should consider in order to obtain 
a practically useful quantum signature scheme to sign quantum messages. 
Nonetheless, our work provides 
a perfect characterization of the AQS schemes yielding forgeable quantum messages,  
although we have taken into account the specific AQS schemes.
Therefore, it could be helpful to conduct further research works related to the AQS, 
and could be also useful to improve quantum cryptographic theories.

%%%%%%%%%%%%%%%%%%%%%%%%%%%%%%%%%%%%%%%%%%%%%%%%%%%%%%%%%%%%%%%%%%%%%%
%%%                                                                %%%
%%%                       Acknowledgements                         %%%
%%%                                                                %%%
%%%%%%%%%%%%%%%%%%%%%%%%%%%%%%%%%%%%%%%%%%%%%%%%%%%%%%%%%%%%%%%%%%%%%%
\begin{acknowledgements}
This research was supported by 
Basic Science Research Program through the National Research Foundation of Korea (NRF) 
funded by the Ministry of  Science, ICT \& Future Planning (NRF-2016R1A2B4014928), 
and T.~Kim was supported by the Ministry of Science, ICT \& Future Planning (NRF-2013R1A1A2063279).
\end{acknowledgements}

%\newpage

\end{document}